\documentclass[a4paper,10pt]{article}
\usepackage[utf8x]{inputenc}
\usepackage{amsfonts,amsthm,amssymb,amsmath}
\usepackage{mathrsfs}
\usepackage{graphicx,graphics}
\usepackage{relsize}
\usepackage[all]{xy}
\usepackage{tikz}
\usepackage{titlesec}
\usepackage{hyperref}
\usepackage[text={7in,9in},margin=1in]{geometry}
\usepackage{authblk}

\allowdisplaybreaks

\numberwithin{equation}{section}
\titleformat*{\section}{\large\bfseries}

\newcommand{\CC}{\mathbb C}

\newcommand{\NN}{\mathbb N}

\newcommand{\RR}{\mathbb R}

\newcommand{\ZZ}{\mathbb Z}

\newcommand{\Hi}{\mathscr{H}}

\newtheorem{thm}{Theorem}[section]
\newtheorem{lemma}[thm]{Lemma}
\newtheorem{cor}[thm]{Corollary}

\newtheorem{rem}[thm]{Remark}

\newcommand{\dprod}[2]{\left\langle #1,#2\right\rangle}
\newcommand{\norm}[1]{\left\lVert #1\right\rVert}
\newcommand{\bkt}[1]{\left(#1\right)}

\title{Jak\v{s}i\'{c}-Last Theorem for Higher Rank Perturbations}
\author{Anish Mallick\thanks{anishm@imsc.res.in}}
\affil{{\small The Institute of Mathematical Sciences,\\Chennai-600113, India}}

\begin{document}
\maketitle
\begin{abstract}
We consider the generalized Anderson Model $\Delta+\sum_{n\in\mathcal{N}}\omega_n P_n$, where $\mathcal{N}$ is a countable set, $\{\omega_n\}_{n\in\mathcal{N}}$ are i.i.d random variables and $P_n$ are rank $N<\infty$ projections. For these models we prove theorem analogous to that of Jak\v{s}i\'{c}-Last on the equivalence of the trace measure $\sigma_n(\cdot)=tr(P_nE_{H^\omega}(\cdot)P_n)$ for $n\in\mathcal{N}$ a.e $\omega$. Our model covers the dimer and polymer models.
\end{abstract}

\section{Introduction}
In this paper we address the nature of spectral measure for generalized Anderson type models with single site potentials of higher rank or a constant randomness over several neighboring collection of sites . The basic setup of the problem is the following. We have a self adjoint operator $A$ on separable Hilbert space $\Hi$, and rank $N$ projections $\{P_n\}_{n\in\mathcal{N}}$ where $\mathcal{N}$ is countable or a finite set. Given an absolutely continuous measure $\mu$ on $\RR$, we define the set of operators
\begin{equation}\label{OpEq1}
H^\omega=A+\sum_{n\in\mathcal{N}}\omega_n P_n
\end{equation}
for $\{\omega_n\}_{n\in\mathcal{N}}\in\mathbb{R}^\mathcal{N}$ distributed identically and independently following the distribution $\mu$. This defines a map from measure space $(\Omega,\mathscr{B},\mathbb{P})$ (product measure space $(\mathbb{R}^\mathcal{N},\mathscr{B}(\mathbb{R}^\mathcal{N}),\otimes\mu)$) to the set of essentially self adjoint linear operators on $\Hi$. We are interested in the spectral measure of this set of operators.

In the case of Anderson tight-binding model, we have the Hilbert space $l^2(\mathbb{Z}^d)$ on which we  have the operator $\Delta$ defined by
$$(\Delta u)(n)=\sum_{|n-m|=1}u(m)\qquad\forall u\in l^2(\ZZ^d),n\in\ZZ^d$$
and the collection of rank one projection $\{|\delta_n\rangle\langle\delta_n|\}_{n\in\mathbb{Z}^d}$. Prior works \cite{JL1,JL2,BS2,BS1} proved simplicity of spectrum for such models using the property that $\{|\delta_n\rangle\langle\delta_n|\}_{n\in\ZZ^d}$ are rank one.

Similar to the tight-binding model, we have the random Schr\"{o}dinger operators, defined by
$$(H^\omega f)(x)=-\sum_{i=1}^d\frac{\partial^2 f}{\partial x^2_i}(x)+\sum_{p\in\ZZ^d}\omega_pG(x-p)f(x)\qquad\forall x\in\RR^d,f\in C_c^\infty(\RR^d)$$
where $G$ is a compactly supported function on $[0,1]^d$. Simplicity of the singular spectrum for this model is still an open problem.

These models are also considered on graphs (for example Bethe lattices and one dimensional strips). In the case of Bethe lattice and Bethe strips absolute continuous spectrum was shown to exists \cite{FHS,K1,KS}. All these models show localization at high disorder \cite{AM} and so have pure point spectrum also.

Multiplicity of these spectra are not well understood for projection valued perturbation. Known results study rank one perturbation \cite{JL1,JL2} and cyclicity \cite{BS1,BS2}. In work by Naboko, Nichols and Stolz \cite{NNS}, such a problem is handled for pure point part of the spectrum. Sadel and Schulz-Baldes \cite{SSH} were looking at quasi-one-dimensional stochastic dirac operator, and provided conditions for singular and absolutely continuous spectrum in terms of size of fibres. They also proved that the multiplicity of ac spectrum is $2$. In this paper we prove results similar to those in \cite{JL1} and \cite{JL2}.

Before stating the main results, we introduce some notations. For $n\in\mathcal{N}$ and $\omega\in\Omega$ define $\Hi_{n,\omega}$ as the cyclic subspace generated by $H^\omega$ defined by \eqref{OpEq1} and the vector space $P_n\Hi$ (this vector space is isomorphic to $\CC^N$), and set $Q_n^\omega:\Hi\rightarrow\Hi_{n,\omega}$ as the canonical projection. Let $E_{H^\omega}$ be the spectral projection measure for the operator $H^\omega$; set $\Sigma^\omega_n(\cdot)=P_nE_{H^\omega}(\cdot)P_n$
(which is now a matrix valued measure) and set $\sigma_n^\omega(\cdot)=tr(\Sigma_n^\omega(\cdot))$ as the trace measure (these are finite measures). Let $P^\omega_{ac}$ be the orthogonal projection onto the absolutely continuous spectral subspace $\Hi_{ac}(H^\omega)$. For $n,m\in\mathcal{N}$, define
\begin{equation}\label{CorEq1}
E_{n,m}=\{\omega\in\Omega|\ Q_n^\omega P_m\text{has same rank as }P_m\}
\end{equation}
We will be working with the following set
$$\mathscr{M}=\{n\in\mathcal{N}|\ \sigma^\omega_n\text{ is not equivalent to Lebesgue measure for a.e }\omega\}$$
The reason for confining oneself in this set is a theorem of F. and M. Riesz \cite{RR}, which implies that \emph{the Borel transform of any complex measure which is zero in $\CC^{+}$ has to be absolutely continuous with respect to Lebesgue measure}.  But one can prove that the total variation measure need to be equivalent to Lebesgue measure, see \cite[Theorem 2.2]{JL2} for a proof. So confining to $\mathscr{M}$ we get that the Borel transform of  non-zero measure in $\mathscr{M}$ can never be identically zero in $\CC^{+}$, and so one can use results about boundary values of analytic functions like \cite{BR2,BR1}. We state the main theorem:

\begin{thm}\label{MainThm}
For any $N\in\NN$, let $\{P_n\}_{n\in\mathcal{N}}$ be collection of rank $N$ projections such that $\sum_{n\in\mathcal{N}}P_n=I$, and $\mu$ be a absolutely continuous measure on $\RR$. Let $\{H^\omega\}_{\omega\in\Omega}$ be a family of operator defined as in \eqref{OpEq1}, then
\begin{enumerate}
\item  For $n,m\in\mathscr{M}$ we have $\mathbb{P}(E_{n,m})\in\{0,1\}$.
\item Let $n,m\in\mathscr{M}$ such that $\mathbb{P}(E_{n,m}\cap E_{m,n})=1$. For a.e $\omega\in\Omega$, the restrictions $P^\omega_{ac}H^\omega|_{\Hi_{n,\omega}}$ and $P^\omega_{ac}H^\omega|_{\Hi_{m,\omega}}$ are unitary equivalent.
\item Let $n,m\in\mathscr{M}$ such that $\mathbb{P}(E_{n,m}\cap E_{m,n})=1$, for a.e $\omega\in\Omega$ the measures $\sigma^\omega_n$ and $\sigma^\omega_m$ are equivalent as Borel measures.
\end{enumerate}
\end{thm}
\begin{rem} Two examples for which the condition $\mathbb{P}(E_{n,m})=1$ can be verified are:
\begin{enumerate}
 \item Consider $l^2(\ZZ)$ with the operator $H^\omega=\Delta+\sum_{n\in\ZZ}\omega_n P_n$ where $P_n=\sum_{k=0}^{N-1}\delta_{Nn+k}$, we have $\mathbb{P}(E_{n,m})=1$ for each $n,m\in\ZZ$. This is because for $n_1,n_2\in\ZZ$, $\dprod{\delta_{n_1}}{(H^\omega)^{|n_1-n_2|}\delta_{n_2}}=1$ and hence all cyclic subspaces intersect with each other non-trivially. If we considered $l^2(\ZZ^d)$ and some enumeration of its basis $\{\delta_{n_k}\}_{k\in\ZZ}$ and define $P_n$ as before, again we can prove $\mathbb{P}(E_{n,m})=1$.
\item Consider the Hilbert space $\oplus_{i=1}^N l^2(\ZZ)$, and $\Delta$ as adjacency operator on each space separately. Set $\pi_i:\ZZ\rightarrow\ZZ$ be surjective map for $i=1,\cdots,N$, and define $P_n=\sum_{i=1}^N \delta^i_{\pi_i(n)}$, where $\{\delta_n^i\}_{n\in\ZZ}$ is basis for each $l^2(\ZZ)$. Then for this case also $\mathbb{P}(E_{n,m})=1$.
\end{enumerate}
In case the measure $\mu$ has compact support on $\RR$ and $A$ is bounded, none of the $\sigma^\omega_n$ can have full support on $\mathbb{R}$, and so $\mathscr{M}=\mathcal{N}$ similar to the rank $1$  case of Jak\v{s}i\'{c} and Last \cite{JL2}.
\end{rem}
The approach to  gain information about the spectral measure is by using the matrix valued function:
$$P_n(H^\omega-z)^{-1}P_n:P_n\Hi\rightarrow P_n\Hi$$
for $z\in\CC^{+}$. Since we will be working with $n\in\mathscr{M}$, it is enough to look at the above matrices. These are termed Matrix valued Herglotz functions or Birman-Schwinger operators. Birman-Schwinger principle was developed for compact perturbations in \cite{BMS,SJ} and some notable applications can be found in \cite{CS,KM,BS3}.

We will be working with the above as Matrix valued Herglotz functions whose properties can be found in \cite{GT1}. By combining theorem \ref{FgetThm1} (see \cite[Theorem 5.4]{GT1} for proof) and \ref{SpecThm} we obtain conditions in terms of $\lim_{\epsilon\downarrow0}P_n(H^\omega-E-\iota\epsilon)^{-1}P_n$.

Second and third part of the theorem \ref{MainThm} are consequences of perturbations by two projections, and the first part is because of Kolmogorov $0$-$1$ law. Lemma \ref{InvLem2} is the primary step for the first part of the main theorem. It tells us that the event $E_{n,m}$ (\emph{$Q^\omega_nP_m$ has same rank as $P_m$}), is independent of any other perturbation, whence Kolmogorov $0$-$1$ law applies. For second part, whenever the condition is satisfied, we have to show that for $E$ in a
full measure set, the density of the measure has same rank for both indices; this is done in corollary \ref{AbsEqu1}. For the last part, the second part of the theorem \ref{MainThm} helps by asserting that absolute continuous parts are equivalent and for the singular part we only need to consider the lowest (Hausdorff) dimensional part. This is the case because we are using Poltoratskii's theorem \cite{POL1}, and lowest dimensional part of the spectrum contributes maximum rate of growth to the
Herglotz function as its argument approaches the boundary of $\CC^{+}$ . Corollary \ref{SingLem3} gives the equivalence for the lowest dimensional parts of the measure.

Before attempting to handle the problem, it is important to note that the \emph{set of perturbations where the procedure may not be applicable is a measure zero set}. Lemma \ref{MbleVar1} gives such a statement, and also tells us that for almost all perturbation, the measure of singular part (w.r.t to Lebesgue measure) is zero.


\section{Preliminaries}
Following lemma is a result concerning the zero sets of polynomials. This lemma helps in the proof of our main theorem by ensuring that for almost all perturbation the set where singular part lie is measure zero.
\begin{lemma}\label{MbleVar1}
For a $\sigma$-finite positive measure space $(X,\mathscr{B},m)$, and a collection of measurable functions $a_i:X\rightarrow\CC$, define the function $f(\lambda,x)=1+\sum_{n=1}^N \lambda^n a_n(x)$. The set defined by
\begin{equation}\label{solset1}
\Lambda_f=\{\lambda\in\CC|m\{x\in X| f(\lambda,x)=0\}>0\}
\end{equation}
is countable.
\end{lemma}
\begin{proof}
The proof is by induction on degree of $f$ (as a polynomial of $\lambda$). We will use the notation:
\begin{equation}\label{MV1Eq1}
S_\lambda=\{x\in X| f(\lambda,x)=0\} 
\end{equation}
By definition the sets $S_\lambda$ are measurable.

Base case of induction is $N=1$, so $f(\lambda,x)=1+\lambda a_1(x)$. Clearly for $\lambda_1\neq \lambda_2\in\CC$ we have $S_{\lambda_1}\cap S_{\lambda_2}=\phi$. Since, if $x\in S_{\lambda_1}\cap S_{\lambda_2}$ then
\begin{align*}
&1+\lambda_1 a_1(x)=0\ \&\ 1+\lambda_2 a_1(x)=0\\
\Rightarrow\qquad &\frac{1}{\lambda_1}=-a_1(x)=\frac{1}{\lambda_2}\\
\Rightarrow\qquad &\lambda_1=\lambda_2
\end{align*}
but we assumed $\lambda_1\neq\lambda_2$. Since $(X,m)$ is $\sigma$-finite, we have a countable collection $\{X_i\}_{i\in\NN}$ such that $\cup_i X_i=X$ and for each $i$ we have $m(X_i)<\infty$. Now for each $\lambda\in\CC$ and $n\in\NN$ define $S_{\lambda,n}=S_\lambda\cap X_n$, so we have $\cup_n S_{\lambda,n}=S_\lambda$, and $\cup_{\lambda\in\Lambda_f}S_{\lambda,n}\subset X_n$. We have
\begin{align*}
\sum_{\lambda\in\Lambda_f}m(S_{\lambda,n})=m(\cup_{\lambda\in\Lambda_f}S_{\lambda,n})\leq m(X_n)<\infty,
\end{align*}
so only for countably many $\lambda\in\Lambda_f$ we have $m(S_{\lambda,n})\neq 0$. Set $\Lambda_n=\{\lambda\in\Lambda_f| m(S_{\lambda,n})>0\}$, we have $\Lambda_f=\cup_{n\in\NN}\Lambda_n$, but since countable union of countable set is countable, we get $\Lambda_f$ countable. This completes base case.

Now assume the induction hypothesis, i.e for measurable functions $a_i:X\rightarrow\CC$, and $f(\lambda,x)=1+\sum_{n=1}^{N}\lambda^n a_n(x)$, the set $\Lambda_f$ is countable.

We have to show for $f(\lambda,x)=1+\sum_{n=1}^{N+1}\lambda^n a_n(x)$, the set $\Lambda_f$ is countable. First we define the relation $\sim$ for elements of $\Lambda_f$; for $\mu,\nu\in\Lambda_f$ we define $\mu\sim\nu$ if there exists $\{\lambda_i\}_{i=1}^k$ such that $\lambda_1=\mu$, $\lambda_k=\nu$ and $m(S_{\lambda_i}\cap S_{\lambda_{i+1}})>0$ for $i=1,\cdots,k-1$.  For $\mu\in\Lambda_f$ we have $\mu\sim\mu$ because $m(S_\mu)>0$ hence $\sim$ is reflexive. If $\mu\sim\nu$ for $\mu,\nu\in\Lambda_f$, then we have a sequence $\{\lambda_i\}_{i=1}^k$ such that $\lambda_1=\mu$ and $\lambda_k=\nu$ and $m(S_{\lambda_i}\cap S_{\lambda_i+1})>0$, hence choosing $\tilde{\lambda}_i=\lambda_{k-i+1}$ we get $\nu\sim \mu$ and so $\sim$ is symmetric. If $\mu\sim\nu$ and $\nu\sim\eta$, then we have sequences $\{\alpha_i\}_{i=1}^p$ and $\{\beta_i\}_{i=1}^q$ such that $\alpha_1=\mu$, $\alpha_p=\beta_1=\nu$ and $\beta_q=\eta$, so defining the sequence $\{\lambda_i\}_{i=1}^{p+q}$ defined as $\lambda_i=\alpha_i$ for $i\leq p$ 
and $\lambda_{i}=\beta_{i-p}$ for $i>p$ we get $\mu\sim\eta$ giving transitivity of $\sim$.
So $\sim$ is a equivalence relation on $\Lambda_f$, and can break the set $\Lambda_f$ into equivalence classes indexed by  $\tilde{\Lambda}=\Lambda_f/\sim$, where we view $[\lambda]\in\tilde{\Lambda}$ as $[\lambda]=\{\mu\in\Lambda_f|\mu\sim\lambda\}$ and define $S_{[\lambda]}=\cup_{\mu\in[\lambda]}S_\mu$.

First we will show for any $[\lambda]\in\tilde{\Lambda}$, the set $[\lambda]$ is countable. Let $\lambda\in\Lambda_f$, so we have the $m(S_\lambda)\neq0$. We will restrict to subspace $S_\lambda$, on this space $f(\nu,x)$ can be written as $f(\nu,x)=\frac{1}{\lambda}(\lambda-\nu)\left(1+\sum_{n=1}^N \tilde{a}_n(x)\nu^n\right)$ (since $\lambda$ is a solution). So we have the new function $\tilde{f}(\nu,x)=1+\sum_{n=1}^N \tilde{a}_n(x)\nu^n$, and by our assumption (induction hypothesis) we get $\Lambda_{\tilde{f}}$ is countable. For any $\nu\in\Lambda_f$ with $m(S_\lambda\cap S_\nu)\neq 0$ implies $\nu\in\Lambda_{\tilde{f}}$, so for fixed $\lambda\in\Lambda_f$ the set of $\nu\in\Lambda_f$ such that $m(S_\lambda\cap S_\nu)\neq 0$ is countable.

Next choose $\lambda\in\Lambda_f$, and set $A_0=\{\lambda\}$, and define
$$A_i=\cup_{\beta \in A_{i-1}}\{\nu\in\Lambda_f| m(S_\nu\cap S_\beta)\neq 0\}\qquad\forall i\in\NN$$
by previous step each $A_{i}$ are countable. So $\cup_{i=0}^\infty A_{i}$ is countable. By definition of $\sim$ we have $[\lambda]=\cup_{i=0}^\infty A_{i}$.

Now we will prove $\tilde{\Lambda}$ is countable. By definition $m(S_{[\lambda]})>0$ for $[\lambda]\in\tilde{\Lambda}$, and for $[\lambda]\neq [\mu]\in\tilde{\Lambda}$ we have $m(S_{[\lambda]}\cap S_{[\nu]})=0$. For $n\in\NN$ define $S_{[\lambda],n}=S_{[\lambda]}\cap X_n$, then we have
\begin{align*}
\sum_{n\in\tilde{\Lambda}}m(S_{[\lambda],n})=m(\cup_{[\lambda]\in\tilde{\Lambda}}S_{[\lambda],n})\leq m(X_i)<\infty
\end{align*}
From last step only countably many $[\lambda]$ can have $m(S_{[\lambda],n})> 0$. Call $\tilde{\Lambda}_n=\{[\lambda]\in\tilde{\Lambda}| m(S_{[\lambda],n})>0\}$ (which are countable); for any $[\lambda]\in\tilde{\Lambda}$ we have
$$0<m(S_{[\lambda]})\leq \sum_{n\in\NN}m(S_{[\lambda],n})$$
So $[\lambda]\in \tilde{\Lambda}$ for some $n\in\NN$ we have $m(S_{[\lambda],n})>0$, hence $\tilde{\Lambda}=\cup_{n\in\NN}\tilde{\Lambda}_n$; giving us $\tilde{\Lambda}$ is countable.

Since $\Lambda_f=\cup_{[\lambda]\in\tilde{\Lambda}}[\lambda]$ and both the sets are countable we get the countability of $\Lambda_f$.

\end{proof}
\begin{rem}
It should be clear that above result holds for a function of the type $f(\lambda,x)=\sum_{n=0}^N a_n(x)\lambda^n$ on the set $\{x\in X|a_0(x)\neq 0\}$. One should note that one cannot extend the result for whole of $X$.

We can view $f(\lambda,x)=\lambda^N\left(\sum_{n=0}^N a_{N-n}(x)\left(\frac{1}{\lambda}\right)^n\right)$, and so the result also holds on the set $\{x\in X|a_N(x)\neq 0\}$.
\end{rem}
\begin{cor}\label{MbleVar2}
For a $\sigma$-finite positive measure space $(X,\mathscr{B},m)$ and a collection of functions $a_i:X\rightarrow\CC$, $b_i:X\rightarrow\CC$, define the function $f(\lambda,x)=\frac{1+\sum_{i=1}^N a_i(x)\lambda^i}{1+\sum_{i=1}^N b_i(x)\lambda^i}$, then the set
\begin{equation}\label{solset2}
\Lambda_f=\{\lambda\in\CC|m\{x\in X| f(\lambda,x)=0\}\neq0\}
\end{equation}
is countable
\end{cor}
\begin{proof}
Set $g(\lambda,x)=1+\sum_{n=1}^N a_n(x)\lambda^n$, then $\{(x,\mu)\in X\times\CC| f(\lambda,x)=0\}\subseteq \{(x,\mu)\in X\times\CC| g(\lambda,x)=0\}$. So by lemma \ref{MbleVar1} we get the desired result.

\end{proof}
We will need the spectral averaging result (see\cite[Corollary 4.2]{CH} for proof):
\begin{lemma}\label{SpecAvaLem}
Let $E_\lambda(\cdot)$ be the spectral family for the operator $A_\lambda=A+\lambda P$, where $A$ is self adjoint operator, and $P$ is a rank $N$ projection. Then for $M\subset\RR$ such that $|M|=0$ (Lebesgue measure), we have $PE_\lambda(M)P=0$ for Lebesgue almost all $\lambda$.
\end{lemma}
This lemma guarantees us that we can omit any Lebesgue measure zero set from any analysis that follows. Following lemma from \cite[Proposition 2.1]{JL1} will be used extensively, as it guarantees the existence  of limits, almost surely. We denote $\Hi_\phi$ to be the cyclic subspace generated by $A$ and $\phi\in\Hi$.
\begin{lemma}\label{NonZeroTms1}
Let $A$ be a self adjoint operator on a separable Hilbert space $\Hi$ with $\phi,\psi\in\Hi$ such that $\Hi_\phi\not\perp\Hi_\psi$. Then for a.e $E\in\RR$ (Lebesgue) the limit
$$\lim_{\epsilon\downarrow 0}\dprod{\phi}{(A-E-\iota\epsilon)^{-1}\psi}=\dprod{\phi}{(A-E-\iota 0)^{-1}\psi}$$
exists and is non-zero.
\end{lemma}
We note that the limit always exists a.e $E$, and it is non-zero if and only if $\Hi_\phi\not\perp\Hi_\psi$.
We will need Poltoratskii's theorem \cite{POL1}.
\begin{thm}\label{PolThm}
For any complex valued Borel measure $\mu$ on $\RR$ and for $f\in L^1(\RR,d\mu)$, with Borel transform $F_\mu(z)=\int\frac{d\mu(z)}{x-z}$
$$\lim_{\epsilon\rightarrow 0}\frac{F_{f\mu}(E+\iota\epsilon)}{F_{\mu}(E+\iota\epsilon)}=f(E)$$
for a.e $E$ with respect to $\mu$-singular.
\end{thm}
A proof can be found in \cite{JL3}. This theorem will be used for proof of equivalence of measure for the singular part in lemma \ref{SingPtSpec1} and corollary \ref{SingLem3}.
\section{Perturbation by finite Rank Projection}
In this section we will be working with $(A,\Hi,\{P_i\}_{i=1}^3)$, where $A$ is a self adjoint operator on the Hilbert space $\Hi$, and $\{P_i\}_{i=1}^3$ are three rank $N$ projections. We will work with the case that the measures $tr(P_iE_A(\cdot)P_i)$ are not equivalent to lebesgue measure (hence using Riesz theorem \cite{RR}, the Borel transform of these measures are non-zero on the upper half plane). Define $A_\mu=A+\mu P_1$, $G_{ij}(z)=P_i(A-z)^{-1}P_j$ and $G^\mu_{ij}(z)=P_i(A_\mu-z)^{-1}P_j$ for $i,j=1,2,3$ and $z\in\CC^{+}$, and will use the notation
$$g(E+\iota0):=\lim_{\epsilon\downarrow0}g(E+\iota\epsilon)$$
for $E\in\RR$ (whenever the limit exists). Using the relation $A^{-1}-B^{-1}=B^{-1}(B-A)A^{-1}=A^{-1}(B-A)B^{-1}$, we have
\begin{align}
&G^\mu_{11}(z)=G_{11}(z)(I+\mu G_{11}(z))^{-1}\label{MHEq1}\\
&(I+\mu G_{11}(z))(I-\mu G^\mu_{11}(z))=I\label{MHEq3}\\
&G^\mu_{ij}(z)=G_{ij}(z)-\mu G_{i1}(z)(I+\mu G_{11}(z))^{-1}G_{1j}(z)\qquad(i,j)\neq(1,1)\label{MHEq2}
\end{align}
For any $E\in\RR$ such that $G_{11}(E+\iota0)$ exists and finite,  and $f:(0,\infty)\rightarrow\CC$ be such that $\lim_{\epsilon\rightarrow 0}f(\epsilon)=0$ look at $\lim_{\epsilon\downarrow 0} f(\epsilon) G_{11}^\mu(E+\iota\epsilon)$ using equation \eqref{MHEq3}
\begin{align*}
&\lim_{\epsilon\downarrow 0} f(\epsilon)(I-\mu G^\mu_{11}(E+\iota\epsilon))(I+\mu G_{11}(E+\iota\epsilon))-f(\epsilon) I=0\\
&(I+\mu G_{11}(E+\iota 0))\left(\lim_{\epsilon\downarrow 0}f(\epsilon) G^\mu_{11}(E+\iota\epsilon)\right)=0
\end{align*}
So we get
\begin{equation}\label{RkNuEq1}
range\left(\lim_{\epsilon\downarrow 0}f(\epsilon) G^\mu_{11}(E+\iota\epsilon)\right)\subseteq ker(I+\mu G_{11}(E+\iota 0))\subseteq ker(\Im G_{11}(E+\iota 0))
\end{equation}
Since $\Im G_{11}(E+\iota 0)\geq 0$ it decomposes the space $P_1\Hi=ker(\Im G_{11}(E+\iota 0))\oplus ker(\Im G_{11}(E+\iota 0))^\bot$ with $range(\Im G_{11}(E+\iota 0))=ker(\Im G_{11}(E+\iota 0))^\bot$, so on $ker(\Im G_{11}(E+\iota 0))^\bot$ we have $\Im G_{ii}(E+\iota 0)>0$. This fact will be used in identifying appropriate subspaces.
We will need some preliminary results before we attempt to prove our main results. The Following lemma  relates the invertibility of the matrices $G^\mu_{12}(z)$ with the ranks of $Q_1P_2$ and $P_2$.
\begin{lemma}\label{InvLem1}
Let $A$ be a self-adjoint operator on a Hilbert space $\Hi$ and  $P_1$ and $P_2$ be two projections of rank $N$. Let $\Hi_i$ denote the cyclic subspace generated by $A$ and $P_i\Hi$ and $Q_i:\Hi\rightarrow\Hi_i$ be the canonical projection onto that subspace, for $i=1,2$. If $Q_1P_2$ has same rank as $P_2$, then $P_1(A-z)^{-1}P_2$ is almost surely invertible for a.e $z\in\CC^{+}$.
\end{lemma}
\begin{proof}
Let $\phi\in P_2\Hi\setminus\{0\}$. Since $Q_1P_2$ has same rank as $P_2$, we have $0\neq Q_1\phi\in\Hi_1$ (if it is zero, then $ker(Q_1)\cap P_2\Hi\neq\{0\}$ and so $rank(Q_1 P_2)<rank(P_2)$),  so there is $\psi\in P_1\Hi$ and $f\in L^2(\RR,d\mu_\psi)$ such that $Q_1\phi=f(A)\psi$. So
\begin{align*}
0\neq\dprod{Q_1\phi}{Q_1\phi}=\dprod{\psi}{f^\ast(A)Q_1\phi}=\dprod{\psi}{f^\ast(A)\phi}=\int \bar{f}(x)d\mu_{\psi,\phi}(x)
\end{align*}
since $Q_1$ commutes with any functions of $A$. So the measure $\mu_{\psi,\phi}$ is non-zero, hence the Borel transform
$$\int\frac{d\mu_{\psi,\phi}(x)}{x-z}=\dprod{\psi}{(A-z)^{-1}\phi},$$
is almost surely non-zero on $\CC^{+}$.

So for each vector $\phi\in P_2\Hi$ there exists a $\psi\in P_1\Hi$ such that $\dprod{\psi}{(A-z)^{-1}\phi}$ is non-zero, in other words $P_1(A-z)^{-1}P_2$ is an injection, and since $P_1(A-z)^{-1}P_2$ is an $n\times n$ matrix we get invertibility.

\end{proof}
\begin{rem}
The lemma above also assures that for almost all $E$ the matrix valued function $P_1(A-E-\iota 0)^{-1}P_2$ is invertible.

For some $z\in\CC^{+}$, the invertibility of $P_1(A-z)^{-1}P_2$ gives us $Q_1P_2$ has same rank as $P_2$. So by looking at $\det(G_{mn}(z))$ we can obtain a statement about non-orthogonality of the subspace $\{\Hi_i\}_{i=1,2}$.
\end{rem}
Choose a basis of $P_i\Hi$, then $G_{ij}(z)$ is a matrix in that basis. We can write
\begin{equation}\label{WorkSet1}
\hspace{-0.3cm}S=\{E\in\RR|\ \text{Entries of $G_{ij}(E+\iota 0)$ exists and are finite } \forall i,j=1,2,3\}
\end{equation}
Then by lemma \ref{NonZeroTms1} we know that $S$ has full measure. Define 
\begin{equation}\label{WorkSet2}
S_{ij}=\{E\in S| G_{ij}(E+\iota 0) \text{ is invertible}\}\qquad\forall i,j=1,2,3
\end{equation}
By lemma \ref{InvLem1}, $S_{ij}$ has full measure whenever $Q_iP_j$ has same rank as $P_j$.
\begin{rem}\label{InvRem1}
On the set $S$, the limit $G_{11}(E+\iota 0)$ exists, and since $\det(I+\mu G_{11}(E+\iota 0))=1+\sum_{i=1}^N a_i(E)\mu^i$, using lemma \ref{MbleVar1} for almost all $\mu$ the matrix $I+\mu G_{11}(E+\iota 0)$ is invertible on a set of full measure. 
\end{rem}

\begin{lemma}\label{InvLem2}
Let $A$ be self adjoint operator on Hilbert space $\Hi$ and $\{P_i\}_{i=1}^3$ be rank $N$ projections. Define $A_\mu=A+\mu P_1$, $G_{ij}(z)=P_i(A-z)^{-1}P_j$ and $G^\mu_{ij}(z)=P_i(A_\mu-z)^{-1}P_j$. If $G_{23}(E+\iota 0)$ is invertible for a.e $E$, then $G^\mu_{23}(E+\iota 0)$ is also invertible for a.e $(E,\mu)$.
\end{lemma}
\begin{proof}
From equations \eqref{MHEq1} and \eqref{MHEq2} and remark \ref{InvRem1} we get
$$G^\mu_{23}(E+\iota 0)=G_{23}(E+\iota 0)-\mu G_{21}(E+\iota 0)(I+\mu G_{11}(E+\iota 0))^{-1}G_{13}(E+\iota 0)$$
since we are only looking for invertibility, looking at determinant is enough. And so 
$$\det(G^\mu_{23}(E+\iota 0))=\frac{\det(G_{23}(E+\iota 0))+\sum a_n(E)\mu^n}{\det(I+\mu G_{11}(E+\iota 0))}$$
Again by corollary \ref{MbleVar2} we get that for almost all $\mu$ the matrix $G_{23}(E+\iota 0)$ is invertible on a set of full measure.

\end{proof}
Next lemma provide the relation between the absolute continuous component of the measure.
\begin{lemma}\label{AbsPtSpec1}
On Hilbert space $\Hi$ we have two rank $N$ projections $P_1,P_2$ and a self adjoint operator $A$. Set $A_\mu=A+\mu P_1$, $G_{ij}(z)=P_i(A-z)^{-1}P_j$ and $G_{ij}^\mu(z)=P_i(A_\mu-z)^{-1}P_j$; set $S$ and $S_{12}$ as \eqref{WorkSet1},\eqref{WorkSet2}. Define 
$$V_{E,i}^\mu=ker(\Im G_{ii}^\mu(E+\iota 0))^{\bot}$$
for each $E\in S\cap\{x\in\RR|\ \lim_{\epsilon\downarrow0} G_{11}^\mu(x+\iota\epsilon)\ \text{exists and finite}\}$. Assume $S_{12}$ has full measure. Then for a.e $\mu$
$$(G_{12}(E+\iota 0))^{-1}:V_{E,1}^\mu\rightarrow V_{E,2}^\mu$$ 
is injective, and 
$$(I+\mu G_{11}(E+\iota 0)):V_{E,1}^0\rightarrow V_{E,1}^\mu$$
is isomorphism.
\end{lemma}
\begin{proof}
From the equation \eqref{MHEq2} and \eqref{MHEq3} we get
$$G^\mu_{22}(z)=G_{22}(z)-\mu G_{21}(z)G_{12}(z)+\mu^2 G_{21}(z)G_{11}^\mu(z)G_{12}(z)$$
For $E\in S\cap \{x\in\RR|\ \lim_{\epsilon\downarrow0} G_{11}^\mu(x+\iota\epsilon)\ \text{exists and finite}\}$, let $v\in V_{E,1}^\mu$, and set $\phi=(G_{12}(E+\iota 0))^{-1}v$, observe (every quantity in RHS below exists and finite so limit can be taken)
\begin{align*}
\lim_{\epsilon\downarrow 0}\dprod{\phi}{(\Im G_{22}^\mu(E+\iota\epsilon))\phi}&=\lim_{\epsilon\downarrow 0} \left[ \dprod{\phi}{(\Im G_{22}(E+\iota\epsilon))\phi}-\mu\dprod{\phi}{\Im( G_{21}(E+\iota\epsilon)G_{12}(E+\iota \epsilon))\phi}\right.\\
&\qquad+\left.\mu^2\dprod{\phi}{(\Im G_{21}(E+\iota\epsilon)G^\mu_{11}(E+\iota \epsilon)G_{12}(E+\iota \epsilon))\phi}\right]
\end{align*}
Since $\Im G_{22}^\mu(E+\iota 0)$ is positive matrix, looking at $\dprod{\phi}{(\Im G_{22}^\mu(E+\iota 0))\phi}$ is enough.

If $\dprod{\phi}{(\Im G_{22}(E+\iota0))\phi}=0$ which implies $(\Im G_{22}(E+\iota0))\phi=0$ so using \eqref{eq4} we have $G_{12}(E+\iota 0)\phi=G_{21}^\ast(E+\iota 0)\phi$, and so
\begin{align*}
\lim_{\epsilon\downarrow 0}\dprod{\phi}{(\Im G_{22}^\mu(E+\iota\epsilon))\phi}&=\mu^2\dprod{G_{12}(E+\iota 0)\phi}{(\Im G^\mu_{11}(E+\iota 0))G_{12}(E+\iota 0)\phi}\\
&\qquad-\mu\dprod{\phi}{\Im( G_{21}(E+\iota0)G_{12}(E+\iota 0))\phi}\\
&=\mu^2\dprod{v}{(\Im G_{11}^\mu(E+\iota 0))v}
\end{align*}
So $\phi\in V^\mu_{E,2}$ and hence $G_{12}(E+\iota 0)^{-1}$ gives the injection.

For the other assertion, let $v\in V_{E,1}^0$ observe
$$\dprod{v}{(I+\mu G_{11}(E+\iota 0))v}=\norm{v}_2^2+\mu(\dprod{v}{\Re G_{11}(E+\iota 0)v}+\iota \dprod{v}{\Im G_{11}(E+\iota 0)v})$$
since $\dprod{v}{\Im G_{11}(E+\iota 0)v}\neq 0$, so the above equation cannot be zero for any $\mu\in \RR$. So on $V_{E,1}^0$ the operator $(I+\mu G_{11}(E+\iota 0))$ is invertible. Set $\phi=(I+\mu G_{11}(E+\iota 0))v$, observe
\begin{align*}
\hspace{-1cm}\lim_{\epsilon\rightarrow 0}\dprod{\phi}{(\Im G_{11}^\mu(E+\iota\epsilon))\phi}&=\lim_{\epsilon\rightarrow 0}\dprod{\phi}{\Im (G_{11}(E+\iota\epsilon)(I+\mu G_{11}(E+\iota\epsilon))^{-1})\phi}\\
\hspace{-1cm}&=\dprod{(I+\mu G_{11}(E+\iota0))^{-1}\phi}{(\Im G_{11}(E+\iota 0))(I+\mu G_{11}(E+\iota0))^{-1}\phi}\\
\hspace{-1cm}&=\dprod{v}{(\Im G_{11}(E+\iota 0))v}\neq0 
\end{align*}
This gives the isomorphism $(I+\mu G_{11}(E+\iota 0)):V_{E,1}^0\rightarrow V_{E,1}^\mu$.
\end{proof}
This only gives the injection between the absolutely continuous spectral subspaces. One cannot expect more from this setting. By a second perturbation we obtain an isomorphism, which is attained in the next corollary.
\begin{cor}\label{AbsEqu1}
Let $A$ be self adjoint operator on Hilbert space $\Hi$, and $P_1,P_2$ are two rank $N$ projections. Set $A_\mu=A+\mu_1 P_1+\mu_2 P_2$ and $G_{ij}(z)=P_i(A-z)^{-1}P_j$, $G^{\mu_1,\mu_2}_{ij}(z)=P_i(A_{\mu_1,\mu_2}-z)^{-1}P_j$ for $i,j=1,2$ and define the vector space
$$V_{E,i}^{\mu_1,\mu_2}=ker(\Im G_{ii}^{\mu_1,\mu_2}(E+\iota 0))^\bot$$
for each $E\in S\cap \{x\in\RR|\ \lim_{\epsilon\downarrow0} G_{ii}^{\mu_1,\mu_2}(x+\iota\epsilon)\ \text{exists and finite for }i=1,2\}$. Assume $S_{12}$, $S_{21}$ have full measure. Then for a.e $\mu_1,\mu_2$ the two vector space $V_{E,1}^{\mu_1,\mu_2}$ and $V_{E,2}^{\mu_1,\mu_2}$ are isomorphic.
\end{cor}
\begin{proof}
This is just application of lemma \ref{AbsPtSpec1}. For $E$ in full measure set we have
\begin{align*}
V_{E,2}^{\mu_1,\mu_2}\hookrightarrow V_{E,1}^{\mu_1,\mu_2}
\end{align*}
where the map is $(G_{21}^{\mu_1,0}(E+\iota 0))^{-1}$. Lemma \ref{InvLem2} tells us $G_{21}^{\mu_1,0}(E+\iota 0)$ is also invertible for almost all $\mu_1$. Now we can do the same thing other way around:
\begin{align*}
V_{E,1}^{\mu_1,\mu_2}\hookrightarrow V_{E,2}^{\mu_1,\mu_2}
\end{align*}
Since we are working in finite dimensional spaces ($V_{E,i}^{\mu_1,\mu_2}$ are finite dimensional), injection in both direction tells us that they are isomorphic.
\end{proof}
The next lemma is similar to lemma \ref{AbsPtSpec1}, but for the singular part. The conclusion is for subspaces where growth of the Herglotz function is maximum or equivalently its associated measure has lowest (Hausdorff) dimension.   We will use the fact that a matrix valued measure $\Sigma_n(\cdot)=P_n E_A(\cdot) P_n$ is absolutely continuous with respect to the trace measure $\sigma_n(\cdot)=tr(\Sigma_n(\cdot))$ and so $\lim_{\epsilon\downarrow 0} \frac{1}{\sigma_n(E+\iota\epsilon)}\Sigma_n(E+\iota\epsilon)=M(E)$ is $L^1$ w.r.t $\sigma_n$-singular ($\sigma_n(z),\Sigma_n(z)$ are corresponding Borel transform).
\begin{lemma}\label{SingPtSpec1}
On Hilbert space $\Hi$ we have two rank $N$ projections $P_1,P_2$ and a self adjoint operator $A$. Set $A_\mu=A+\mu P_1$, $G_{ij}(z)=P_i(A-z)^{-1}P_j$ and $G_{ij}^\mu(z)=P_i(A_\mu-z)^{-1}P_j$. Set $f_E(\epsilon)=tr(G_{11}^\mu(E+\iota\epsilon))^{-1}$ and $E\in\RR$ be such that $f_E(\epsilon)\xrightarrow{\epsilon\downarrow 0}0$, define
$$\tilde{V}_{E,i}^\mu=ker\left(\lim_{\epsilon\downarrow 0}f_E(\epsilon)G_{ii}^\mu(E+\iota\epsilon)\right)^\bot$$
Assume $S_{12}$ defined as \eqref{WorkSet2} has full measure, then for $E\in S$ such that $f_E(\epsilon)\xrightarrow{\epsilon\downarrow 0}0$  defined as in \eqref{WorkSet1} the map
$$(G_{12}(E+\iota 0))^{-1}: \tilde{V}^\mu_{E,1}\rightarrow \tilde{V}^\mu_{E,2}$$
is injective. So the measure $\sigma_2^\mu$ (where $\sigma^\mu_i(\cdot)=tr\left(P_i E_{A_\mu}(\cdot)P_i\right)$) is absolutely continuous with respect to $\sigma_1^\mu$-singular.
\end{lemma}
\begin{proof}
Using $i,j=2$ in \eqref{MHEq2}, we have
$$G_{22}^\mu(z)=G_{22}(z)-\mu G_{21}(z)G_{12}(z)+\mu^2 G_{21}(z)G_{11}^\mu(z)G_{12}(z)$$
Since we are working with $E\in S$, the limits for $G_{ij}(E+\iota 0)$ exists for $i,j=1,2$. For $\phi,\psi\in P_2\Hi$ we have
\begin{align*}
\dprod{\psi}{G_{22}^\mu(E+\iota\epsilon)\phi}&=\dprod{\psi}{G_{22}(E+\iota\epsilon)\phi}-\mu \dprod{\psi}{G_{21}(E+\iota\epsilon)G_{12}(E+\iota\epsilon)\phi}\\
&\qquad+\mu^2 \dprod{\psi}{G_{21}(E+\iota\epsilon)G_{11}^\mu(E+\iota\epsilon)G_{12}(E+\iota\epsilon)\phi}\\
\lim_{\epsilon\downarrow0}f_E(\epsilon)\dprod{\psi}{G_{22}^\mu(E+\iota\epsilon)\phi}&=\mu^2\lim_{\epsilon\downarrow 0}f_E(\epsilon)\dprod{\psi}{G_{21}(E+\iota\epsilon)G_{11}^\mu(E+\iota\epsilon)G_{12}(E+\iota\epsilon)\phi}\\
&=\mu^2 \dprod{\psi}{G_{21}(E+\iota0)\left(\lim_{\epsilon\downarrow0}f_E(\epsilon) G_{11}^\mu(E+\iota\epsilon)\right)G_{12}(E+\iota 0)\phi}
\end{align*}
And now using \eqref{RkNuEq1} and \eqref{eq4} we have
\begin{align*}
& \dprod{\psi}{G_{21}(E+\iota0)\left(\lim_{\epsilon\downarrow0}f_E(\epsilon) G_{11}^\mu(E+\iota\epsilon)\right)G_{12}(E+\iota 0)\phi}\\
&\qquad= \dprod{\psi}{G_{12}(E+\iota0)^\ast\left(\lim_{\epsilon\downarrow0}f_E(\epsilon) G_{11}^\mu(E+\iota\epsilon)\right)G_{12}(E+\iota 0)\phi}
\end{align*}
From above if $\phi= G_{12}(E+\iota 0)^{-1} v$ for $v\in \tilde{V}_{E,1}^\mu$, then $\phi\in \tilde{V}^\mu_{E,2}$, giving us that the map $G_{12}(E+\iota 0)^{-1}$ is injection.

Finally
$$\lim_{\epsilon\downarrow 0}\frac{tr\left(G^\mu_{22}(E+\iota\epsilon)\right)}{tr\left(G^\mu_{11}(E+\iota\epsilon)\right)}=tr\left(G_{12}(E+\iota0)^\ast\left(\lim_{\epsilon\downarrow0}f_E(\epsilon) G_{11}^\mu(E+\iota\epsilon)\right)G_{12}(E+\iota 0)\right)$$
where RHS is $L^1$ for $\sigma^\mu_1$-singular by lemma \ref{PolThm} (Poltoratskii's theorem).
\end{proof}
\begin{cor}\label{SingLem3}
Let $A$ be self adjoint operator on Hilbert space $\Hi$, and $P_1,P_2$ are two rank $N$ projections. Set $A_\mu=A+\mu_1 P_1+\mu_2 P_2$, $G_{ij}(z)=P_i(A-z)^{-1}P_j$ and $G^{\mu_1,\mu_2}_{ij}(z)=P_i(A_{\mu_1,\mu_2}-z)^{-1}P_j$ for $i,j=1,2$. Let $E\in S_{12}\cap S_{21}$ (defined as in \eqref{WorkSet2}) and $tr(G_{ii}^{\mu_1,\mu_2}(E+\iota\epsilon))^{-1}\xrightarrow{\epsilon\downarrow0}0$ for either $i=1,2$, then
$$\tilde{V}_{E,i}^{\mu_1,\mu_2}=ker(\lim_{\epsilon\downarrow 0}tr(G_{ii}^{\mu_1,\mu_2}(E+\iota\epsilon))^{-1}G_{ii}^{\mu_1,\mu_2}(E+\iota\epsilon))^\bot\qquad i=1,2$$
are isomorphic. In particular the singular part of trace measure associated with $G_{ii}^{\mu_1,\mu_2}$ are equivalent to each other.
\end{cor}
\begin{proof}
Define
$$\tilde{V}_{E,i,j}^{\mu_1,\mu_2}=ker(\lim_{\epsilon\downarrow 0}tr(G_{jj}^{\mu_1,\mu_2}(E+\iota\epsilon))^{-1}G_{ii}^{\mu_1,\mu_2}(E+\iota\epsilon))^\bot$$
This is exactly like corollary \ref{AbsEqu1}. By action of lemma \ref{SingPtSpec1} we have
$$V^{\mu_1,\mu_2}_{E,1,1}\hookrightarrow V^{\mu_1,\mu_2}_{E,2,1}\ \&\ V^{\mu_1,\mu_2}_{E,2,2}\hookrightarrow V^{\mu_1,\mu_2}_{E,1,2}$$
where first is given by $G_{12}^{0,\mu_2}(E+\iota 0)^{-1}$ and second is given by $G_{21}^{\mu_1,0}(E+\iota 0)^{-1}$ which are a.e (with respect to perturbation $\mu_1,\mu_2$) invertible because of lemma \ref{InvLem2}. Because of the second conclusion of the previous lemma \ref{SingPtSpec1} we have 
$$\lim_{\epsilon\downarrow 0}\frac{tr\left(G^\mu_{11}(E+\iota\epsilon)\right)}{tr\left(G^\mu_{22}(E+\iota\epsilon)\right)}\ \ \text{exists for a.e $tr(P_2 E_{A_\mu}(\cdot)P_2)$-singular},$$
$$\lim_{\epsilon\downarrow 0}\frac{tr\left(G^\mu_{22}(E+\iota\epsilon)\right)}{tr\left(G^\mu_{11}(E+\iota\epsilon)\right)}\ \ \text{exists for a.e $tr(P_1 E_{A_\mu}(\cdot)P_1)$-singular}.$$
So as a vector space $V_{E,i,j}^{\mu_1,\mu_2}=V_{E,i,i}^{\mu_1,\mu_2}=V^{\mu_1,\mu_2}_{E,i}$ for a.e $tr(P_i E_{A_\mu}(\cdot)P_i)$-singular.
Since we have injection both direction and finite dimensionality of the spaces involved, we get the isomorphism.
\end{proof}
\subsection{Proof of Main theorem}
\begin{proof}
The notation we will use is
$$G^\omega_{nm}(z)=P_n(H^\omega-z)^{-1}P_m\qquad\forall n,m\in\mathcal{N}$$
and for some $p\in\mathscr{M}$ we will denote
$$H^{\omega}_{\mu,p}=H^\omega+\mu P_p$$
and
$$G^{\omega,\mu,p}_{nm}(z)=P_n(H^\omega_{\mu,p}-z)^{-1}P_m\qquad\forall n,m\in\mathscr{M}$$
\begin{enumerate}
\item For $n,m\in\mathscr{M}$, let $\omega\in E_{n,m}$, using lemma \ref{InvLem1} we get $G^\omega_{nm}(z)$ is almost surely invertible. For any $p\in\mathcal{N}$, we have $H^\omega_{\mu,p}$, and using lemma \ref{InvLem2} we get $G^{\omega,\mu,p}_{nm}(z)$ is also almost surely invertible for almost all $\mu$. So we get, if $\omega\in E_{n,m}$, then $\tilde{\omega}\in E_{n,m}$ ($\tilde{\omega}$ is defined by $\omega_n=\tilde{\omega}_n\ \forall n\in\mathscr{M}\setminus\{p\}$) or in other words, $E_{n,m}$ is independent of the $\omega_p$ for any $p\in\mathscr{M}$. We can repeat the procedure and show that $E_{n,m}$ is independent of $\{\omega_{p_i}\}_{i=1}^K$ for $p_i\in\mathscr{M}$ . So we can use Kolmogorov $0$-$1$ law to conclude that $\mathbb{P}(E_{n,m})\in\{0,1\}$.
\item For any $n\in\mathscr{M}$, we have $(H^\omega,\Hi_{n,\omega})$ is unitary equivalent to $(M_{id},L^2(\RR,\Sigma^\omega_n,\CC^N))$ (see theorem \ref{SpecThm}). For $m\in\mathscr{M}$ such that $\mathbb{P}(E_{n,m}\cap E_{m,n})=1$, we have to show $(\Sigma^\omega_n)_{ac}$ is equivalent to $(\Sigma^\omega_m)_{ac}$. Using $(5)$ of theorem \ref{FgetThm1} we have
$$d(\Sigma^\omega_n)_{ac}(E)=\frac{1}{\pi}\Im G_{nn}^\omega(E+\iota 0)dE$$
For $\omega\in E_{n,m}$, we can write the operator $H^{\tilde{\omega}}=H^\omega+\mu_1 P_n+\mu_2 P_m$, and using corollary \ref{AbsEqu1} we get $V_n^{{\tilde{\omega}}}$ are isomorphic to $V_m^{{\tilde{\omega}}}$, where 
$$V_i^{{\tilde{\omega}}}=ker\left(P_i(H^{\tilde{\omega}}-E-\iota 0)^{-1}P_i\right)^\bot$$
Since $\Im G^\omega_{nn}(E+\iota 0)=\Im\left(P_n(H^\omega-E-\iota0)^{-1}P_n\right)$, the isomorphism gives the equivalence. By proof of part $(1)$, we know $E_{n,m}$ is independent of $\omega_n$ and $\omega_m$, so the result holds for a.e $\omega$.
\item For $n,m\in\mathscr{M}$ such that $\mathbb{P}(E_{n,m}\cap E_{m,n})=1$. Let $\omega\in E_{n,m}$, define $H^{\tilde{\omega}}=H^\omega+\mu_n P_n+\mu_m P_m$ (almost always $\tilde{\omega}\in E_{n,m}$), then corollary \ref{SingLem3}  gives the equivalence of the trace measure for singular part. As for absolute continuous part, second part of the theorem gives the equivalence.
\end{enumerate}

\end{proof}

\section*{Acknowledgement}
I would like to thank M. Krishna for discussions and helpful suggestions. This work is partially supported by IMSc Project 12-R\&D-IMS-5.01-0106.
\appendix\section{Appendix}
For $A\in M_n(\CC)$, we have the decomposition $A=\Re A+\iota \Im A$, where both $\Re A$ and $\Im A$ are self adjoint. For $A_{11},A_{12},A_{21},A_{22}\in M_n(\CC)$ such that $A_{11}^\ast A_{22}-A_{21}^\ast A_{12}=I$ and $A_{21}^\ast A_{11}=A_{11}^\ast A_{21}$, $ A_{22}^\ast A_{12}=A_{12}^\ast A_{22}$. Define
$$\tilde{A}=(A_{11}-A_{12}A)(A_{21}-A_{22}A)^{-1}$$
then we have
\begin{align}\label{eq2}
\Im\tilde{A}&=\frac{1}{2\iota}(\tilde{A}-\tilde{A}^\ast)\nonumber\\
&=\frac{1}{2\iota}\left((A_{11}-A_{12}A)(A_{21}-A_{22}A)^{-1}-\left((A_{21}-A_{22}A)^{-1}\right)^\ast(A_{11}-A_{12}A)^\ast\right)\nonumber\\
&=\left((A_{21}-A_{22}A)^{-1}\right)^\ast\Im A\left((A_{21}-A_{22}A)^{-1}\right)
\end{align}
This is equivalent to mobius transforms for complex numbers. These kind of transform will play important role for determining $\Sigma_n^\omega$.

Following lemma gives some of the properties of the off-diagonal terms of Herglotz matrices.
\begin{lemma}\label{PP1}
For $A_{ij}\in M_n(\CC)\ i,j=1,2$, define
$$A=\bkt{\begin{matrix} A_{11} & A_{12} \\ A_{21} & A_{22}\end{matrix}}$$
with the property that $\Im A\geq 0$. Then for $u,v\in\CC^n$
\begin{equation}\label{eq3}
\left|\dprod{u}{\frac{A_{12}-A_{21}^\ast}{2\iota}v}\right|^2\leq 2\dprod{u}{(\Im A_{11})u}\dprod{v}{(\Im A_{22})v}
\end{equation}
\end{lemma}
 It's proof follows same steps as $2\times2$ matrix case. As a consequence of \eqref{eq3}, we have
\begin{equation}\label{eq4}
(\Im A_{22})v=0\ \Rightarrow\ A_{12}v=A_{21}^\ast v, \& \ (\Im A_{11})u=0\ \Rightarrow\ A_{21}u=A_{12}^\ast u
\end{equation}
also
\begin{equation}\label{eq5}
\Im tr(A_{22})=0\ \Rightarrow\ A_{12}=A_{21}^\ast, \& \ \Im tr(A_{11})=0\ \Rightarrow\ A_{21}=A_{12}^\ast
\end{equation}

We will use Matrix valued Herglotz function to work, following theorem from \cite[Theorem 5.4]{GT1} list some of the useful properties. Proof of these statements are similar to that of scalar case.
\begin{thm}\label{FgetThm1}
Let $M:\CC^{+}\rightarrow M_n(\CC)$ be a matrix-valued Herglotz function, then
\begin{enumerate}
\item $M(z)$ has finite normal limits, i.e $M(E\pm\iota 0)=\lim_{\epsilon\rightarrow 0}M(E+\iota\epsilon)$ for a.e $E\in\RR$.
\item If for each diagonal element $M_{ii}(z),1\leq i\leq n$ of $M(z)$ has zero normal limit on a fixed subset of $\RR$ which has positive Lebesgue measure, then $M(z)=C_0$ where $C_0$ is a constant self-adjoint $n\times n$ matrix with vanishing diagonal elements.
\item There exists a matrix-valued measure $\Sigma$ on the bounded Borel subset of $\RR$ satisfying
$$\int \dprod{v}{d\Sigma(x)v}(1+x^2)^{-1}\qquad\forall v\in\CC^n$$
such that the Nevanlinna, respectively, Riesz-Herglotz representation
$$M(z)=C+Dz+\int_\RR d\Omega(x)\left((x-z)^{-1}-x(1+x^2)^{-1}\right)\qquad\forall z\in\CC^{+}$$
$$C=M(\iota),\ D=\lim_{\eta\uparrow\infty}\frac{1}{\iota\eta}M(\iota\eta)$$
holds.
\item The Stieltjes inversion formula for $\Sigma$ is
$$\frac{1}{\pi}\lim_{\epsilon\downarrow0}\int_{\lambda_1}^{\lambda_2}d\lambda \Im(M(\lambda+\iota\epsilon))=\frac{\Sigma(\{\lambda_1\})+\Sigma(\{\lambda_2\})}{2}+\Sigma((\lambda_1,\lambda_2))$$
\item The absolute continuous part of the measure is given by
$$d\Sigma_{ac}(\lambda)=\frac{1}{\pi}\Im(M(\lambda+\iota 0))d\lambda$$
\item Any poles of $M(z)$ are simple and are located on real axis. 
\end{enumerate}
\end{thm}
We will use the given version of spectral theorem.
\begin{thm}\label{SpecThm}
Let $A$ be a self adjoint operator on Hilbert space $\Hi$ and $P$ be an rank $N$ projection. Let the vector space $P\Hi$ has basis $\{\delta_n\}_{n=1}^N$ and define the cyclic subspace generated by $A$ and $\delta_n$ by $\Hi_n$ for all $n=1,\cdots,N$, and define the subspace
$$\Hi_P=\sum_{i=1}^N \Hi_{i}$$
also let $\Sigma_A$ denote the spectral projection of $A$, then
 $(L^2(\mathbb{R},P\Sigma_AP,P\Hi),id)$ and $(\Hi_P,A)$ are unitarily equivalent, where $id$ is multiplication by identity.
\end{thm}
\begin{proof}
We have a basis of $P\Hi$ given by $\{\delta_n\}_{n=1}^N$. So define the map
$$U:L^2(\mathbb{R},P\Sigma_AP,\CC^N)\rightarrow \Hi_P$$
as
$$U(f_1,\cdots,f_N)\mapsto\sum_{i=1}^N f_i(A)\delta_i$$
the map is injection because
\begin{align*}
0&=\norm{U(f_1,\cdots,f_n)}_2^2=\sum_{i,j=1}^N\dprod{f_i(A)\delta_i}{f_j(A)\delta_j}\\
&=\sum_{i,j=1}^N \int \bar{f_i}(x)f_j(x)d\mu_{ij}(x)
\end{align*}
where $\mu_{ij}(\cdot)$ is the measure $\dprod{\delta_i}{\Sigma_A(\cdot)\delta_j}$, so the last equation tells us
$$\int\dprod{f(x)}{dP\Sigma_AP(x)f(x)}=0$$
where $f(x)=(f_1(x),\cdots,f_N(x))$, and $(P\Sigma_A(\cdot)P)_{ij}=\dprod{\delta_i}{\Sigma_A(\cdot)\delta_j}$. So the map $U$ is injection. The map is surjection because for $\phi\in\Hi_P$ by definition we can find $\{f_i\}_{i=1}^N$ such that $f_i\in L^2(\mathbb{R},\mu_{ii},\CC)$ such that $\phi=\sum_{i=1}^Nf_i(A)\delta_i$, and so $(f_1,\cdots,f_N)$ maps to $\phi$.
Finally we have to show $UM=AU$, by definition $(Mf)(x)=xf(x)$
\begin{align*}
U(Mf)&=\sum_{i=1}^N (Af_i(A))\delta_i\\
&=A\sum_{i=1}^N f_i(A)\delta_i=A(Uf)
\end{align*}
completing the proof.

\end{proof}
\bibliographystyle{plain}

\end{document}